\documentclass[aps,pre,onecolumn,superscriptaddress]{revtex4}  

\usepackage{graphicx}  
\usepackage{dcolumn}   
\usepackage{bm}        
\usepackage{amssymb,times}   
\usepackage{times}

\usepackage{amssymb}
\usepackage{amsfonts}
\usepackage{amsmath}
\usepackage{graphicx}%
\usepackage{amsthm}
\usepackage{epstopdf}
\usepackage{color}
\usepackage{blkarray}

\providecommand{\U}[1]{\protect\rule{.1in}{.1in}}
\newcommand{\commenting}[1]{\textcolor{black}{#1}} 
\newtheorem{theorem}{Theorem}
\newtheorem{corollary}{Corollary}

\begin{document}

\title{Structural transition in interdependent networks with regular interconnections}

\author{Xiangrong Wang}
\affiliation{Faculty of Electrical Engineering, Mathematics and Computer Science, Delft University of Technology, Delft, The Netherlands.}

\author{Robert E. Kooij}
\affiliation{Faculty of Electrical Engineering, Mathematics and Computer Science, Delft University of Technology, Delft, The Netherlands.}
\affiliation{iTrust Centre for Research in Cyber Security, Singapore University of Technology and Design, Singapore.}

\author{Yamir Moreno}
\affiliation{Institute for Biocomputation and Physics of Complex Systems (BIFI), University of Zaragoza, Zaragoza 50009, Spain}
\affiliation{Department of Theoretical Physics, University of Zaragoza, Zaragoza 50009, Spain}
\affiliation{ISI Foundation, Turin, Italy}

\author{Piet Van Mieghem}
\affiliation{Faculty of Electrical Engineering, Mathematics and Computer Science, Delft University of Technology, Delft, The Netherlands.}

\date{\today}

\begin{abstract}
Networks are often made up of several layers that exhibit diverse degrees of interdependencies. A multilayer interdependent network consists of a set of graphs $G$ that are interconnected through a weighted interconnection matrix $ B $, where the weight of each inter-graph link is a non-negative real number $ p $. Various dynamical processes, such as synchronization, cascading failures in power grids, and diffusion processes, are described by the Laplacian matrix $ Q $ characterizing the whole system. For the case in which the multilayer graph is a multiplex, where the number of nodes in each layer is the same and the interconnection matrix $ B=pI $, being $ I $ the identity matrix, it has been shown that there exists a structural transition at some critical coupling, $ p^* $. This transition is such that dynamical processes are separated into two regimes: if $ p > p^* $, the network acts as a whole; whereas when $ p<p^* $, the network operates as if the graphs encoding the layers were isolated. In this paper, we extend and generalize the structural transition threshold $ p^* $ to a regular interconnection matrix $ B $ (constant row and column sum). Specifically, we provide upper and lower bounds for the transition threshold $ p^* $ in interdependent networks with a regular interconnection matrix $ B $ and derive the exact transition threshold for special scenarios using the formalism of quotient graphs. Additionally, we discuss the physical meaning of the transition threshold $ p^* $ in terms of the minimum cut and show, through a counter-example, that the structural transition does not always exist. Our results are one step forward on the characterization of more realistic multilayer networks and might be relevant for systems that deviate from the topological constrains imposed by multiplex networks.
\end{abstract}

\maketitle

\section{Introduction}

An interdependent network, also called an interconnected network or a network of networks, is a multilayer network consisting of different types of networks that depend upon each other for their functioning \cite{van2016interconnectivity}. The most illustrative example of these systems is perhaps given by multilayer power networks, in which the power system is represented in one layer which in turn is connected to a communication network whose topology is encoded by another layer. The nodes in the former are controlled by those in the second, whereas at the same time the elements of the communication layer need power to function \cite{rosato2008modelling}, closing the feedback between the two graphs. The study of the structure and dynamics of interdependent networks is of utmost importance, as critical infrastructures such as the previous one, and other telecommunications, transportation, water/oil/gas-supply systems, etc, are highly interconnected and mutually depend upon each other. As for the dynamics, the framework of multilayer interdependent networks constitutes a useful approach to address catastrophic events such as large-scale blackouts, whose causes are rooted in the inherent vulnerability associated to the interdependencies between the different components of a complex multilayer system: the failure of one infrastructure propagates to another infrastructure \cite{vespignani2010complex} and so on. Indeed,  Little \cite{little2004socio} has proposed to view critical infrastructures as systems of systems so as to understand their robustness against cascading failures. 

A key aspect of multilayer networks that has received less attention is the coupling between the layers that make up the whole system, as it can modify the outcome of dynamical processes that run on top of them. For example, Buldyrev \emph{et al.} \cite{buldyrev2010catastrophic} showed that the collapse of interdependent networks occurs abruptly while the failure of individual networks is approached continuously. Also, the epidemic threshold for disease spreading processes is characterized by both the topology of each coupled network and the interconnection between them \cite{wang2013effect,cozzo2013,bonaccorsi2014epidemic,arruda2017}. On the other hand, the authors of \cite{radicchi2013abrupt} studied an interdependent model consisting of two connected networks, $ G_1 $ and $ G_2 $, with weighted interconnection links. The coupling weight between the two networks is determined by a non-negative real value $ p $, which, for instance, can be interpreted as the power dispatched by an element of the power layer in the power-communication system above. The previous interdependent system has been shown \cite{radicchi2013abrupt,martin2014algebraic} to exhibit a structural transition that takes place at a coupling value, $ p^* $, that separates two regimes: for $ p>p^* $, the interdependent network acts as a whole, whereas for $ p<p^* $, the network is structurally separated and the layers $ G_1 $ and $ G_2 $ behave as if they were isolated. The explicit expression for the transition threshold $ p^* $ is determined in \cite{sahneh2015exact}.

The model in \cite{radicchi2013abrupt} focuses on a one-to-one interconnection between nodes of different layers. This means that one node in graph $ G_1 $ connects to one and only one node in graph $ G_2 $ and vice versa. When the interconnection pattern is not one-to-one, as in most real-world examples, the determination of the transition threshold $ p^* $ is more complex. Examples of a multiple-to-multiple interconnection pattern can be found in (i) smart grids consisting of coupled sensor networks and power networks \cite{ganesan2006power,sood2009developing,parandehgheibi2013robustness} where a sensor might control multiple power stations due to cost and energy constraints; (ii) functional brain networks modeled as multilayer networks where one brain region in one layer might be functionally connected to any node in another layer \cite{Tewarie2016integrating}; and (iii) infrastructures like power networks and fiber-optic communication systems that are geographically interconnected based on spatial proximity \cite{rinaldi2001identifying,wang2016modeling}. Given the abundance of the previous examples and similar scenarios, it is thus relevant to extend the study of structural transitions to such cases.

In this paper, we investigate the structural threshold $ p^* $ of interdependent networks with a general $ k $-to-$ k $ ($ k $ is a positive integer) interconnection, see Figure \ref{fig_interGraph_model_manytomany} and Section \ref{sec_preliminary}, where we introduce these networks. In Section \ref{sec_upper_bound_p*}, we derive upper and lower bounds for the structural threshold $ p^* $ and report on certain topologies whose exact transition threshold can be calculated from its quotient graph. The physical interpretation of the structural threshold $ p^* $ with respect to the minimum cut is presented in Section \ref{sec_physical_meaning_p*}. Next, in Section \ref{sec_exact_p*}, we derive the exact structural threshold $ p^* $ for special cases of interconnectivity and present a counter example for the non-existence of the structural threshold $ p^* $. Section \ref{sec_conclusion} concludes the paper. 
\begin{figure}[t!]
\centering
\includegraphics[width=0.7\textwidth]{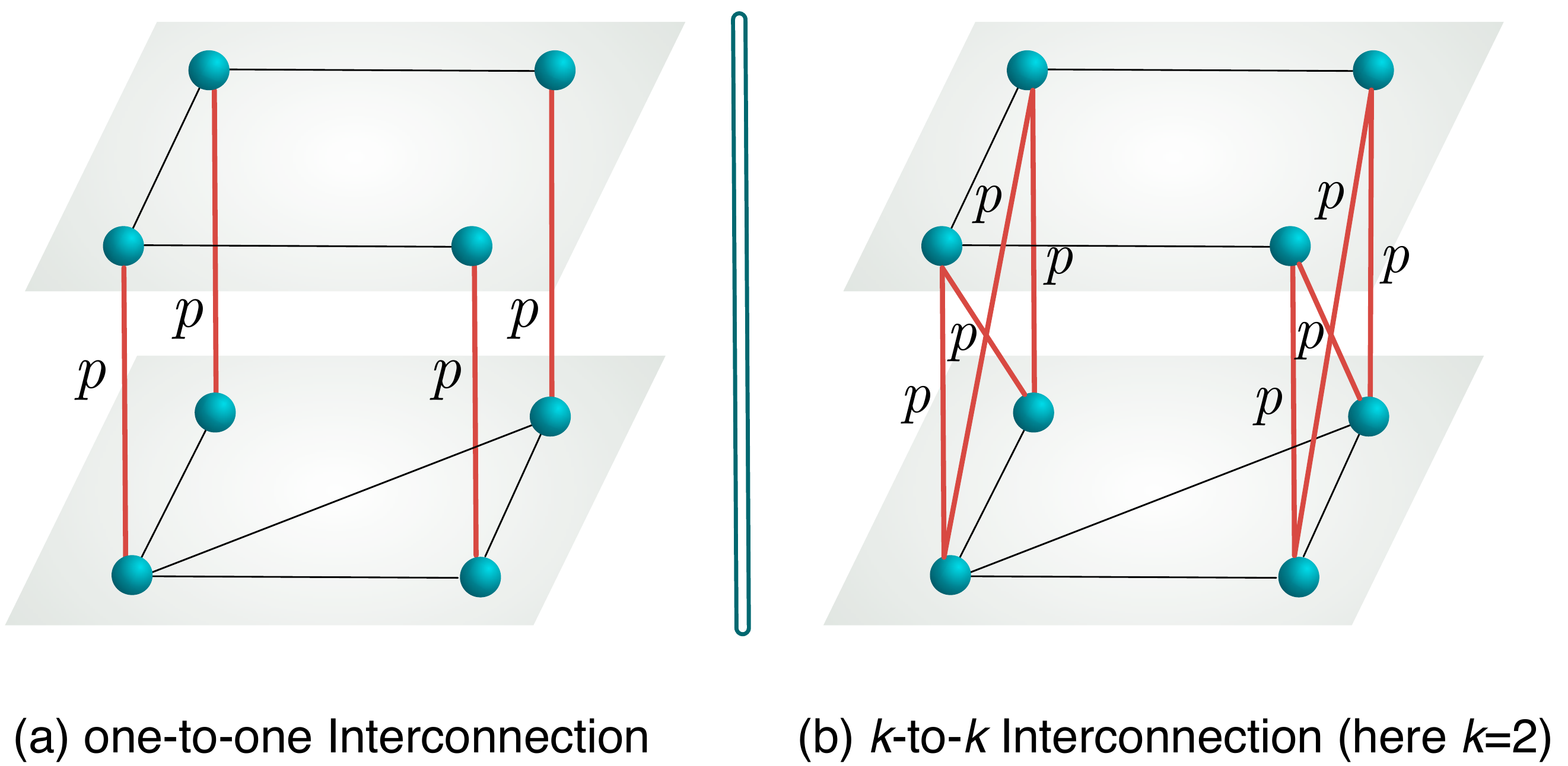}
\caption{The figure depicts two interconnected networks that only differ in their interlayer couplings: the left panel represents the case studied in  \cite{radicchi2013abrupt}, which is however not commonly found in real systems as it constrains all interlayer couplings to follow a one-to-one interconnection pattern. On the contrary, the right panel represents a scenario in which the interlayer connectivity follows a $ k $-to-$ k $ coupling scheme, and although still not fully realistic given the regularity -i.e., homogeneity- of the interconnection pattern, it is more complex and representative of real interconnected networks.}
\label{fig_interGraph_model_manytomany}
\end{figure}

\section{Interdependent networks}
\label{sec_preliminary}
Let the graph $ G(N,L) $ represents an interdependent, multilayer network consisting of two layers (networks), described by graph $ G_1 $ with $ n $ nodes and graph $ G_2 $ with $ m $ nodes. The total number of nodes in $ G $ is thus $ N=n+m $. An interdependent link is that which connects a node $ i $ in network $ G_1 $ to a node $ j $ in network $ G_2 $. The adjacency matrix $ A $ of the interdependent network $ G $ has a block structure of the form,
\begin{equation*}
A=
\begin{bmatrix}
\left(A_1\right)_{n \times n} & B_{n \times m}\\
\left(B^T\right)_{m \times n} & \left(A_2\right)_{m \times m},
\end{bmatrix}
\end{equation*}
where $ A_1 $ is the $ n \times n $ adjacency matrix of $ G_1 $, $ A_2 $ is the $ m \times m $ adjacency matrix of $ G_2 $ and $ B $ is the $ n \times m $ coupling or interconnection matrix encoding the connections between $ G_1 $ and $ G_2 $. If each interdependent link is weighted with a non-negative real number $ p $, the matrix $ B $ is a weighted matrix with elements $ b_{ij}=p $ if node $ i $ in $ G_1 $ connects to node $ j $ in $ G_2 $, otherwise $ b_{ij}=0 $. Note that the definition for $ B $ used in \cite{van2016interconnectivity} is more general, as the weights of each interdependent link can be different. Here, the matrix $ B $ corresponds to an scenario in which each interdependent link has a weight $ p $, the same for all links of this kind.

A \textbf{$ k $-to-$ k $  interconnection}, where $ k=1,\ 2,\ \cdots,\ \min(n,m) $, means that one node in graph $ G_1 $ connects to $ k $ nodes in graph $ G_2 $ and vice-versa. We only consider undirected interconnection links. The $ k $-to-$ k $ interconnection requires a square matrix $ B $ with $ n=m $, because the number $ kn $ of interconnection links computed in graph $ G_1 $ must be equal to the number $ km $ computed in graph $ G_2 $, i.e., $ kn=km $. For the rest of this article, we focus on a square interconnection matrix $ B $ with $ n=m $ and the subscript of matrix $ B $ is omitted. Furthermore, as noted before, the $ k $-to-$ k $ interconnectivity pattern is a generalization of the one-to-one scheme ($ B=p I $) studied in \cite{buldyrev2010catastrophic,radicchi2013abrupt,sahneh2015exact}.

For a square coupling matrix $ B $, a $ k $-to-$ k $ interconnection can be constructed via a circulant matrix \cite{PVM_graphspectra} with the form
\begin{equation}
B=
\begin{bmatrix}
        c_1 & c_2 & c_3 & \cdots & c_{n} \\
c_n & c_1 & c_2 & \cdots & c_{n-1} \\
c_{n-1} & c_n & c_1 & \cdots & c_{n-2} \\
 \vdots  & \vdots  & \vdots & \ddots & \vdots\\
 c_{2} & c_{3} & c_{4}  & \cdots  & c_{1}
\end{bmatrix}
\label{eq_circulant_B}
\end{equation}
where the row vector $ \left(c_1, \ c_2, \ \ldots, \ c_n\right) $ has exactly $ k $ elements equal to $ p $ and $ n-k $ elements that are $ 0 $. A circulant matrix is a matrix where each row is the same as the previous one, but the elements are shifted one position right and wrapped around at the end. Circulant matrices are commutative \cite{davis2012circulant}. For example, a symmetric matrix $ B $ for a $ 2 $-to-$ 2 $ ($ k=2 $) interconnection can be written as
\begin{equation*}
B=
\begin{bmatrix}
0 & p & 0 & \cdots & p \\
 p & 0 & p & \cdots & 0 \\
 0 & p & 0 & \cdots & 0 \\
 \vdots  & \vdots  & \vdots & \ddots & \vdots\\
 p & 0 & 0  & \cdots  & 0\\
\end{bmatrix}
\end{equation*}

Analogous to the definition of the Laplacian matrix $ Q=\Delta-A $ in a single network, where $ \Delta $ is the diagonal matrix of node degrees, we use the following diagonal matrices:
\begin{equation*}
\Delta_1 \overset{def}{=}\text{diag}\left(Bu\right)
\end{equation*}
\begin{equation*}
\Delta_2 \overset{def}{=}\text{diag}\left(B^Tu\right)
\end{equation*}
to define the Laplacian matrix $ Q $ of the interdependent network $ G $ as
\begin{equation*}
Q=
\begin{bmatrix}
Q_1+\Delta_1 & -B\\
-B^T & Q_2+\Delta_2
\end{bmatrix}
\end{equation*}
where $ Q_1 $ and $ Q_2 $ are the Laplacian matrices of networks $ G_1 $ and $ G_2 $, respectively. The all-one vector is denoted by $ u $ and the subscript of $ u $ is used if the dimension is not clear. Since the Laplacian matrix $ Q $ is symmetric, the eigenvalues of $ Q $ are non-negative and at least one is zero \cite{PVM_graphspectra}. We order the eigenvalues of the Laplacian matrix $ Q $ as $ 0=\mu_{N} \leq \mu_{N-1} \leq \cdots \leq \mu_1 $ and denote the eigenvector corresponding to the $ k $-largest eigenvalue by $ x_k $. The second smallest eigenvalue of the Laplacian matrix $ Q $ was coined by Fiedler \cite{fiedler1973algebraic} as the algebraic connectivity $ \mu_{N-1} $ of a graph $ G $. The algebraic connectivity plays a key role in different aspects related to the structure and dynamics of networks, such as in diffusion processes \cite{martin2014algebraic, gomez2013diffusion}, synchronization stability \cite{wu2005synchronization} and network robustness against failures \cite{jamakovic2008robustness}.

The Laplacian eigenvalue equation for the eigenvector $ x_k=\left(x_1^T,\ x_2^T\right)^T $, where $ x_1 $ and $ x_2 $ are $ n \times 1 $ vectors, associated to the eigenvalue $ \mu_k $ is
\begin{equation}
\begin{bmatrix}
Q_1+\Delta_1 & -B\\
-B^T & Q_2+\Delta_2
\end{bmatrix}\begin{bmatrix}
x_1\\
x_2
\end{bmatrix}=\mu_k\begin{bmatrix}
x_1\\
x_2
\end{bmatrix}
\label{eq_laplacian_eigenvalue_equation}
\end{equation}
The normalized vector $ x_N=\frac{1}{\sqrt{N}}\left(u^T_{n},\ u^T_{n}\right)^T $ is an eigenvector associated to the smallest eigenvalue $ \mu_N=0 $ of the Laplacian $ Q $. We briefly present a theorem (in \cite[Theorem 3]{van2016interconnectivity}) to introduce a non-trivial eigenvalue and eigenvector of the Laplacian $ Q $.   

\begin{theorem}
\label{non_trival_eigenmode}
Only if the $ n \times m $ interconnection matrix $ \widetilde{B} $ has a constant row sum equal to $ \frac{\mu^*}{N}m $ and a constant column sum equal to $ \frac{\mu^*}{N}n $, which we call the regularity condition for $ \widetilde{B}_{n \times m} $,
\begin{equation*}
\begin{cases}
\widetilde{B}u_m=\frac{\mu^*}{N}mu_n\\
\widetilde{B}^Tu_n=\frac{\mu^*}{N}nu_m\\
\end{cases}
\end{equation*} 
then
\begin{equation*}
x=\frac{1}{\sqrt{N}}\left[\sqrt{\frac{m}{n}}u^T_{n},\  -\sqrt{\frac{n}{m}}u^T_{m}\right]^T 
\end{equation*}
\commenting{is} an eigenvector of $ Q $ associated to the eigenvalue 
\begin{equation*}
\mu^*=\left(\frac{1}{n}+\frac{1}{m}\right)u_n^T\widetilde{B}_{n \times m}u_m
\end{equation*}
and $ u_n^T\widetilde{B}_{n \times m}u_m $ equals the sum of the elements in $ \widetilde{B} $, representing the total strength of the interconnection between graphs $ G_1 $ and $ G_2 $.
\end{theorem}
\begin{corollary}
\label{corollary_nontrivial_eigenmode}
Consider an interdependent graph $ G $ with $ N $ nodes consisting of two graphs each with $ n $ nodes, whose interconnections are described by a weighted interconnection matrix $ B $. For a $ k $-to-$ k $ interconnection pattern with the coupling weight $ p $ on each interconnection link, the vector
\begin{equation}
x=\frac{1}{\sqrt{N}}\left[u^T_{n},\  -u^T_{n}\right]^T 
\end{equation}
is an eigenvector of the Laplacian matrix $ Q $ of graph $ G $ associated to the eigenvalue 
\begin{equation}
\label{eq_nontrivial_eigenvalue}
\mu^*=2kp
\end{equation} 
\end{corollary}
\begin{proof}
For a $ k $-to-$ k $ interconnection, the row and column sum of the interconnection matrix $ B $ is a constant equal to $ kp $,
\begin{equation*}
\begin{cases}
Bu_n=kpu_n\\
B^Tu_n=kpu_n\\
\end{cases}
\end{equation*}
which obeys the regularity condition in Theorem \ref{non_trival_eigenmode}. With $ n=m $ and the total coupling strength $  u_n^T\widetilde{B}_{n \times m}u_m=kpn $ in Theorem \ref{non_trival_eigenmode}, we establish the Corollary \ref{corollary_nontrivial_eigenmode}.
\end{proof}
Corollary \ref{non_trival_eigenmode} shows the existence of an eigenvalue $  \mu^*=2kp   $. Given that the coupling weight $ p $ on each interconnection link can be varied from $ 0 $ to $ \infty $, there is a value of $ p >0 $ for which $ \mu^*=2kp  $ in (\ref{eq_nontrivial_eigenvalue}) can be made the smallest positive eigenvalue, which then equals the algebraic connectivity $ \mu_{N-1} $ of the whole interdependent network $ G $. By increasing the coupling weight $ p $, the non-trivial eigenvalue $ \mu^*=2kp $ is no longer the second smallest eigenvalue. There exists a transition threshold $ p^* $ such that $ \mu_{N-1} \neq 2kp $ when $ p > p^* $. Because the eigenvalues of the Laplacian $ Q $ are continuous functions of the coupling weight $ p $, the second and third smallest eigenvalue coincide \cite{sahneh2015exact} at the point of the transition threshold $ p^* $. 

Finally, the Laplacian matrix $ Q $ for a $ k $-to-$ k $ interconnection can be written as the sum of two matrices $ Q=\begin{bmatrix}
Q_1 & O\\
O & Q_2
\end{bmatrix}+\begin{bmatrix}
kpI & -B\\
-B^T & kpI
\end{bmatrix} $. 
Moreover, according to the interlacing theorem for the sum of two matrices \cite{PVM_graphspectra}, a lower bound for the third smallest eigenvalue $ \mu_{N-2} $ of the Laplacian matrix $ Q $ follows
\begin{equation}
\mu_{N-2}\left(Q\right) \geq \min(\mu_{n-2}\left(Q_1\right),\ \mu_{n-2}\left(Q_2\right))
\label{eq_lower_bound_third_smallest_eigenvalue}
\end{equation}
where $ \mu_{n-2}\left(Q_1\right) $ and $ \mu_{n-2}\left(Q_2\right) $ are the third smallest eigenvalue of graphs $ G_1 $ and $ G_2 $, respectively. 

\section{Bounds and exact expression for the transition threshold $ p^* $ }
\label{sec_upper_bound_p*}
This section derives both upper and lower bounds for the transition threshold $ p^* $ of interdependent networks with $ k $-to-$ k $ ($ k\geq1 $) interconnection patterns. We find topologies of interdependent networks where an exact analytical expression for the transition threshold can be attained.
\subsection{Upper and lower bounds for $ p^* $}
For a given interconnection matrix $ B $ with a $ k $-to-$ k $ interconnection, i.e., $ Bu=B^Tu=kpu $, the Laplacian matrix $ Q $ can be written as
\begin{equation}
Q= \begin{bmatrix}
Q_1+kpI & -B\\
-B^T & Q_2+kpI
\end{bmatrix}
\label{eq_Q_two-to-two_interconnecion} 
\end{equation}
For any normalized vector $ x=\left(x_1^T,\ x_2^T\right)^T $, the quadratic form $ x^TQx $ of the Laplacian $ Q $ follows
\begin{equation}
x^TQx=kp+x_1^TQ_1x_1+x_2^TQ_2x_2-2x_1^TBx_2
\label{eq_quadraticQ}
\end{equation}

Let $ x_1$ be an eigenvector associated to the second smallest eigenvalue $ \mu_{n-1}(Q_1)$ of $ Q_1 $ and $ x_2=0 $. For the vector $ x=\left(x_1^T,\ 0\right)^T $, its normalization reads $ x^Tx=x_1^Tx_1=1 $. Thus, the quadratic form in (\ref{eq_quadraticQ}) follows $ x^TQx=kp+  \mu_{n-1}(Q_1)$. Analogously, we have $ x^TQx=kp+  \mu_{n-1}(Q_2)$ when $ x_1=0 $ and $ x_2 $ is the eigenvector associated to $ \mu_{n-1}(Q_2)$. Applying the Rayleigh inequality \cite{PVM_graphspectra} to the algebraic connectivity $ \mu_{N-1} $ yields
\begin{equation*}
\mu_{N-1} \leq \frac{x^TQx}{x^Tx}
\end{equation*}
With $ x=\left(x_1^T,\ 0\right)^T  $ or $ x=\left(0,\ x_2^T\right)^T  $, we arrive at
\begin{equation}
\mu_{N-1} \leq \min\left(\mu_{n-1}(Q_1),\ \mu_{n-1}(Q_2)\right)+kp.
\label{eq_algebraic_inequality}
\end{equation}
The previous equality holds when $ x $ is the eigenvector associated to the algebraic connectivity $ \mu_{N-1} $.

Next, note that the non-trivial eigenvalue $ \mu^*=2kp $ in (\ref{eq_nontrivial_eigenvalue}) corresponding to the eigenvector $ x=\frac{1}{\sqrt{N}}\left(u^T_{n},\ -u^T_{n}\right)^T $ is no longer the algebraic connectivity $ \mu_{N-1} $ when $ p>p^* $. At the transition threshold, $ p^* $, the algebraic connectivity is $ \mu_{N-1}=2kp^* $. Substituting $ \mu_{N-1}=2kp^*  $ and $ p=p^* $ in (\ref{eq_algebraic_inequality}), we arrive at an upper bound for the transition threshold $ p^* $
\begin{equation}
p^* \leq \frac{1}{k}\min(\mu_{n-1}(Q_1),\ \mu_{n-1}(Q_2))
\label{eq_upper_bound_treshold}
\end{equation} 

To obtain a lower bound, we apply the min-max theorem to the quadratic term $ x_1^TBx_2 $ in (\ref{eq_quadraticQ}), which yields
\begin{equation*}\label{key}
x_1^TBx_2 \leq \sigma_1 x_1^Tx_2
\end{equation*}
where $ \sigma_1 $ is the largest singular value of the matrix  $ B $, which is $ \sigma_1 = \sqrt{\lambda_1\left(B^TB\right) } = kp$.
According to the Cauchy–-Schwarz inequality, we have that $ x_1^Tx_2 \leq ||x_1||||x_2|| \leq \frac{1}{2} $. Thus, the quadratic form for the Laplacian matrix $ Q $ reads as
\begin{equation*}
x^TQx \geq kp+x_1^TQ_1x_1+x_2^TQ_2x_2-\sigma_1.
\end{equation*}
At the transition point $ p^* $, we have
\begin{equation*}
2kp^*  \geq \min(\mu_{n-1}(Q_1),\ \mu_{n-1}(Q_2))\left( x_1^T x_1+x_2^T x_2\right)
\end{equation*}
With $ \left( x_1^T x_1+x_2^T x_2\right) = 1 $, the transition threshold $ p^* $ is lower bounded by
\begin{equation}
p^* \geq \frac{\min(\mu_{n-1}(Q_1),\ \mu_{n-1}(Q_2))}{2k}
\label{eq_lower_bound_treshold}
\end{equation}

\begin{figure}[!thb]
\centering
\includegraphics[width=0.8\textwidth]{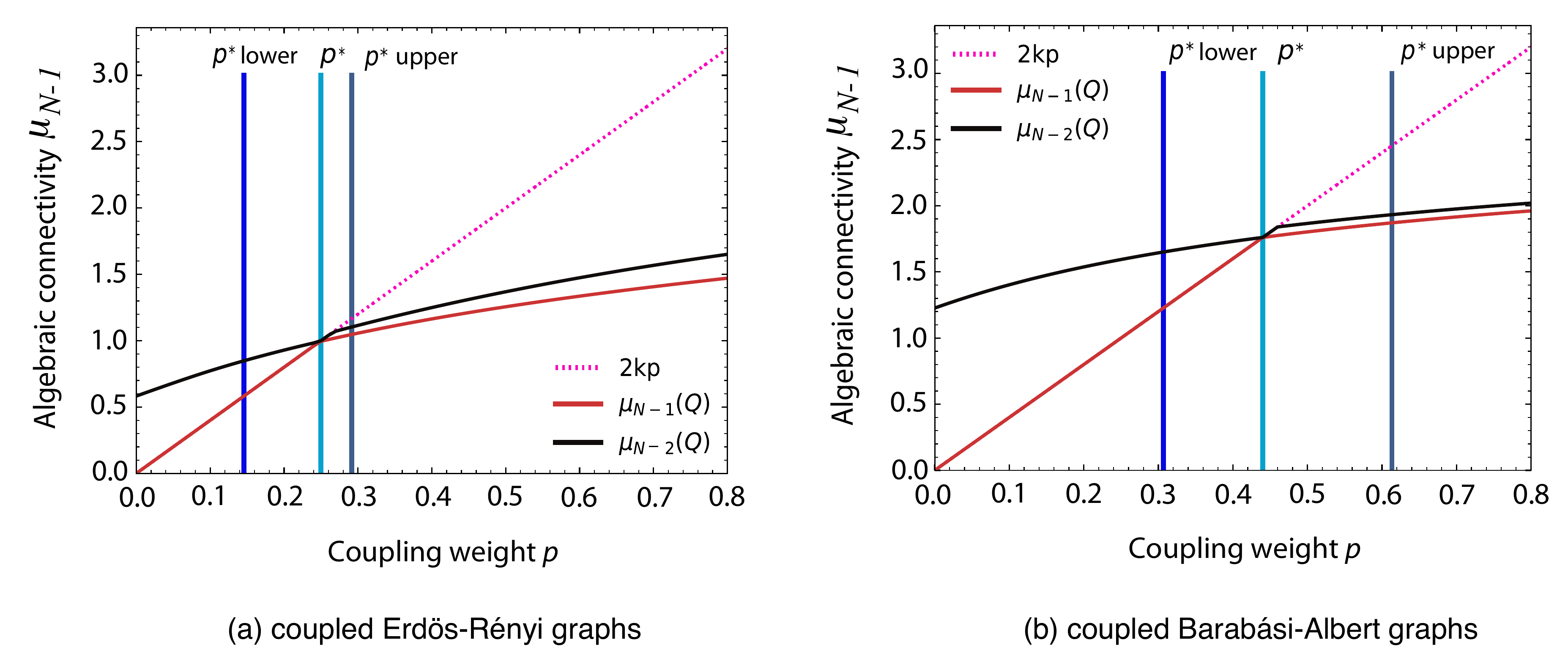}
\caption{Accuracy of the upper and lower bounds for the transition threshold $ p^* $ in interdependent networks consisting of (a) two Erd\H{o}s-R\'{e}nyi graphs $ G_q(n) $ and (b) two Barab\'{a}si-Albert graphs with $ n=500 $ and average degree $ d_{av}=6 $. The interconnection pattern is a $ 2 $-to-$ 2 $ scheme, i.e., $ k=2 $.}
\label{fig_performanceBounds_ER_BA}
\end{figure}

Figure \ref{fig_performanceBounds_ER_BA} shows the accuracy of the upper (\ref{eq_upper_bound_treshold}) and lower  (\ref{eq_lower_bound_treshold}) bounds for interdependent networks of size $ N=1000 $ that consist of two Erd\H{o}s-R\'{e}nyi graphs (panel a) $ G_q(n) $ with link density $ q $, as well as two Barab\'{a}si-Albert graphs (panel b) with average degree $ d_{av}=6 $. The interconnection pattern is a $ 2 $-to-$ 2 $ ($ k=2 $) scheme. As can be seen in the figure, the upper bound is more accurate for the two ER networks, whereas the lower bound seems to work slightly better the other way around. We also note that in \cite{radicchi2013abrupt}, it was shown that the transition threshold $ p^* $ is upper-bounded by $ p^* \leq \frac{1}{4}\mu_{N-1}(Q_1+Q_2) $ when $ B=pI $ (the $ k $-to-$ k $ interconnection with $ k=1 $). For this case, the exact value of $ p^* $ was determined in \cite{sahneh2015exact}, however, the method used can not be readily generalized to a two-to-two nor to a general $ k $-to-$ k $ ($ k \geq 2 $) interconnection pattern. 

\commenting{\subsection{Exact expression using the quotient graph}}
In this subsection, we present an analytical approach to calculate the exact transition threshold for a class of networks. The approach uses partitions of graphs in each layer and the corresponding quotient graph \cite{sanchez2014} whose transition threshold is analytically solvable. 

Let us first focus on the partitions of the graph in each layer. For a $ k $-to-$ k $ interconnection pattern, we divide the graph $ G_1 $ of $ n $ nodes into $ \frac{n}{k} $ subgraphs $ H^{(1)}_1, \ \ldots, \ H^{(1)}_\frac{n}{k} $ where each subgraph has exactly $ k $ nodes and $ \frac{n}{k} $ is an integer. In other words, $ k $ is chosen in such a way that $ k \ | \ n $, i.e., $ k $ is a divisor of $ n $.  Analogously, a similar division is performed on graph $ G_2 $, resulted in the subgraphs $ H^{(2)}_1, \ \ldots, \ H^{(2)}_\frac{n}{k} $. Without loss of generality, we demonstrate the whole approach by using subgraph $ G_1 $ \commenting{and denote $ \frac{n}{k} $ by $ m $}.
After the division, \commenting{subgraphs $ H_i $ are ordered as a chain and each subgraph connects to its neighboring subgraphs}. The adjacency matrix $ A_1 $ of graph $ G_1 $, consisting of divided subgraphs $ H_i $ and \commenting{connected as a chain}, can be written as a block matrix
\begin{equation}
A_1=\begin{bmatrix}
A_{H_1} & R_1 &\\
R_1^T & A_{H_2} & R_2\\
& \ddots & \ddots & \ddots\\
&  & R_{m}^T& A_{H_{m}}\\ 
\end{bmatrix}
\label{eq_A_1}
\end{equation}
where the $ k \times k $ adjacency matrix for a subgraph $ H_i $  is denoted by $ A_{H_i}$.  

In order to calculate the exact transition threshold, we perform a coarse-grained process by condensing each subgraph $ H_i $ into a node and two nodes are connected if two subgraphs are connected. The resulting graph corresponding to the partition is also called the quotient graph \cite{sanchez2014}. The link between nodes $ i $ and $ j $ in the quotient graph is weighted by the average degree $ \widetilde{d_{ij}} $ that a node in subgraph $ H_i $ has in subgraph $ H_j $. \commenting{The adjacency matrix of the quotient graph of $ G_1 $ reads
\begin{equation*}\label{key}
\left(\widetilde{A_1}\right)_{m\times m}=\begin{bmatrix}
\widetilde{d_{11}} & \widetilde{d_{12}}  &&\\
\widetilde{d_{21}}  & \widetilde{d_{22}} & \widetilde{d_{23}} \\
&\ddots & \ddots & \ddots\\
&&  \widetilde{d_{m,m-1}}& \widetilde{d_{mm}}\\ 
\end{bmatrix}
\end{equation*} }
If the $ k $-to-$ k $ interconnection is attained by fully connecting subgraph $ H^{(1)}_i $ in graph $ G_1 $ to the subgraph $ H^{(2)}_i $ in graph $ G_2 $, then the Laplacian of the quotient graph of the whole interdependent network is
\begin{equation}\label{eq_coarse_grained_laplacian}
\widetilde{Q} =\begin{bmatrix}
\left(\widetilde{Q_1}\right)_{m\times m} & -kpI \\
-kpI& \left(\widetilde{Q_2}\right)_{m\times m}\\
\end{bmatrix}
\end{equation}
where 
\begin{equation*}\label{key}
\left(\widetilde{Q_1} \right)_{ij}= \left\{
\begin{split}
-\widetilde{d_{ij}}  & \qquad \text{if} \quad i \neq j \\
\sum_{s = 1, s \neq i}^{m}\widetilde{d_{is}} + kp &\qquad  \text{if}   \quad i= j
\end{split}
\right. 
\end{equation*}
When the quotient Laplacian preserves the algebraic connectivity of the original graph, then by applying the method proposed in \cite{sahneh2015exact} to the quotient Laplacian (\ref{eq_coarse_grained_laplacian}), the transition occurs at 
\begin{equation}\label{eq_topology_exact_threshold}
p^* = \frac{1}{2k}\lambda_{ m-1}\left(\widetilde{Q_1} \left(\frac{\widetilde{Q_1}+\widetilde{Q_2}}{2}\right)^\dagger \widetilde{Q_2}\right)
\end{equation}
where $ \dagger $ denotes the pseudo-inverse \cite{van2017pseudoinverse} and $\lambda_{ m-1}  $ denotes the second smallest eigenvalue. Figure \ref{fig_topology_exact_threshold} shows an example for which the transition threshold is determined by (\ref{eq_topology_exact_threshold}). 

\begin{figure}[t!]
\centering
	\includegraphics[width=0.8\textwidth]{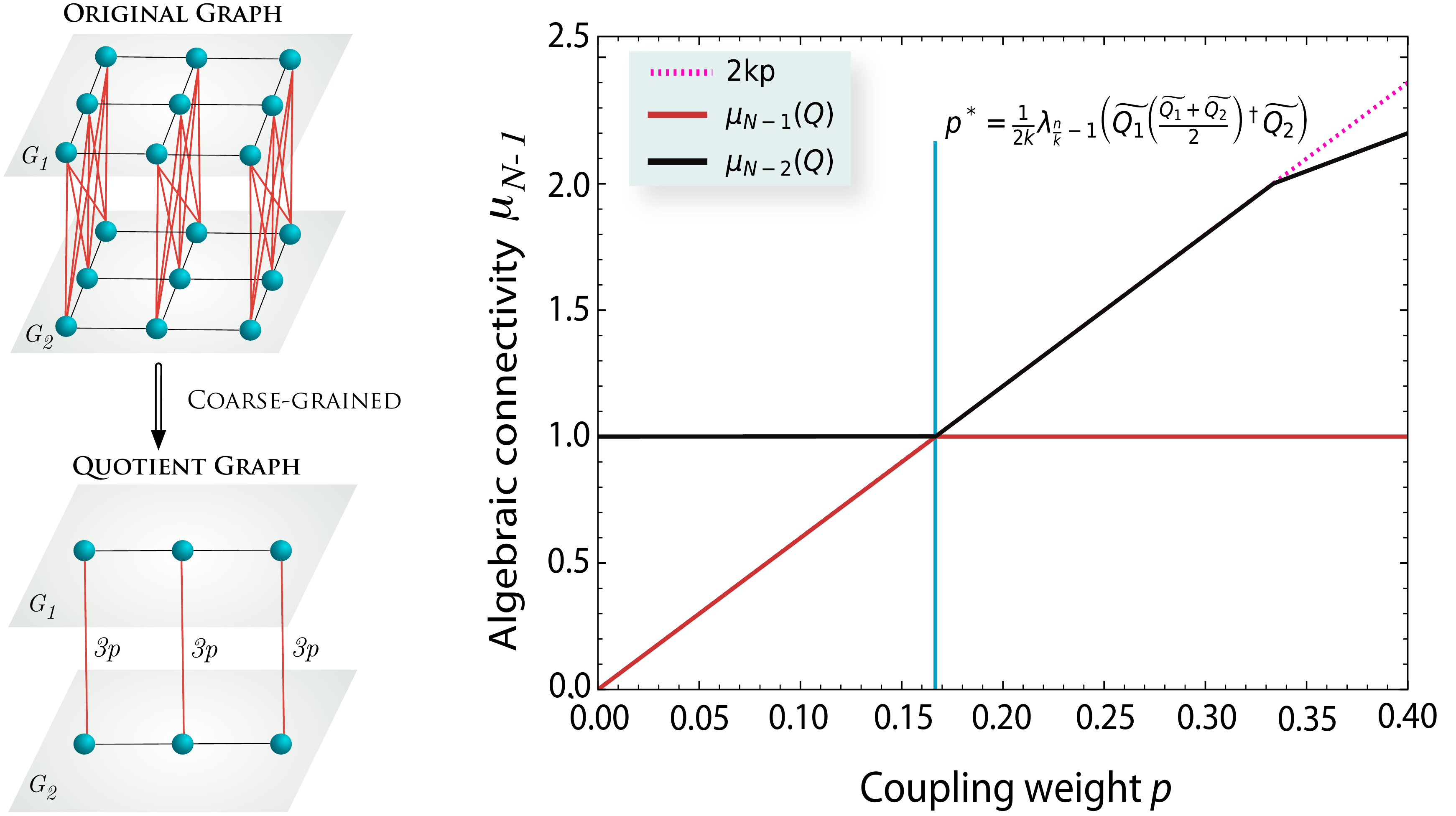}
	\caption{Example topology with $ k $-to-$ k $ ($ k=3 $) interconnections whose transition threshold is calculated exactly by (\ref{eq_topology_exact_threshold}).}
\label{fig_topology_exact_threshold}
\end{figure}

The partition of the two-layer interdependent network into subgraphs $ H_i $ with constant row sum $ R_i $, for any $ i $, is also called an equitable partition. Besides, the eigenvalues of the quotient matrix $ \widetilde{Q}  $ of an equitable partition are also eigenvalues of $ Q $. \commenting{In particular, if all the subgraphs are identical and $ R$ is the identity matrix, then the quotient Laplacian $ \widetilde{Q} $ preserves the algebraic connectivity of the original interdependent network. The proof is given in Appendix}. The approach can be applied to coarse-graining of these particular graphs such that certain Laplacian eigenvalues are preserved, which is arguably a key issue to analyze for large complex networks \cite{gfeller2007spectral,faccin2017entrograms}. Moreover, Cozzo and Moreno \cite{cozzo2016characterization} employed coarse-grained or quotient graphs to characterize multiple structural transitions of coupled (with $k=1$) multilayer networks.
 
\section{Physical meaning of $ p^* $ in terms of the minimum cut}
\label{sec_physical_meaning_p*}
In graph theory, a cut \cite{PVM_graphspectra} is defined as the partition of a graph into two disjoint subgraphs $ \widetilde{G_{1}} $ and $ \widetilde{G_{2}} $. A cut set refers to a set of links between subgraphs $ \widetilde{G_{1}} $ and $ \widetilde{G_{2}} $. For a weighted graph, the minimum cut refers to a cut set whose cut weight $ R $ is minimized, where the cut weight $ R $ is the sum of link weights over all links in the cut set. In this paper, we consider interdependent networks, $ G $, that are weighted, where each link within graphs $ G_1 $ and $ G_2 $ has weight $ 1 $ and each link between graphs $ G_1 $ and $ G_2 $ has weight $ p $. 

A normalized index vector $ y $ for a cut of a graph $ G $ into subgraphs $ \widetilde{G_{1}} $ and $ \widetilde{G_{2}} $ is defined as
\begin{equation*}
y_i=\sqrt{\frac{1}{N}}\begin{cases}
1 & \text{if node } i \in \widetilde{G_1}\\
-1 & \text{if node } i \in \widetilde{G_2}\\
\end{cases}
\end{equation*}
where $ y^Ty=1 $. The cut weight $ R $ follows \cite{PVM_graphspectra} from the quadratic form of the Laplacian matrix $ Q $
\begin{equation*}
R=\frac{Np}{4}\sum_{l \in \mathcal{L}}(y_{l^+}-y_{l^-})^2= \frac{N}{4}y^TQy
\end{equation*}
because $ y_{l^+}-y_{l^-}=\frac{2}{\sqrt{N}} $ if the starting node $ l^+ $ and the ending node $ l^- $ of a link $ l $ belong to different subgraphs, otherwise $ y_{l^+}-y_{l^-}=0 $. 
The minimum cut is \cite{PVM_graphspectra}
\begin{equation*}
R_{\min}=\frac{N}{4}\min_{y \in \mathbb{Y}} y^TQy
\end{equation*}
where $ \mathbb{Y} $ is the set of all possible normalized index vectors of the $ N $-dimensional space. Rayleigh's theorem \cite{PVM_graphspectra} states that, for any normalized vector $ y $ orthogonal to the all-one vector $ u $, we have that $\mu_{N-1} \leq \frac{y^TQy}{y^Ty} \leq y^TQy $ because $ y^Ty=1 $, and the equality holds when $ y $ is an eigenvector associated to $ \mu_{N-1} $. With $\mu_{N-1} \leq y^TQy $, the minimum cut $ R_{\min} $ follows
\begin{equation*}
R_{\min} \geq \frac{N}{4}\mu_{N-1}
\end{equation*}
If the index vector $ y $ is an eigenvector of $ G $ associated to the eigenvalue $ \mu_{N-1} $, then we obtain that $ R_{\min}  =\frac{N \mu_{N-1}}{4}$. On the other hand, Corollary \ref{corollary_nontrivial_eigenmode} implies that the eigenvalue $ \mu^*=2kp $ can be made the second smallest eigenvalue $ \mu_{N-1} $ with eigenvector $x=\frac{1}{\sqrt{N}}\left[u^T_{n},\  -u^T_{n}\right]^T $ if $ p < p^* $. If this regime, the partition corresponding to $ y=x $ results in the minimum cut with $ R_{\min}=\frac{N \mu_{N-1}}{4} $. The resulting subgraphs from that partition are exactly graphs $ G_1 $ and $ G_2 $ and the cut set contains all the interdependent links. Contrarily, when the coupling weight $ p > p^* $, the eigenvector $x=\frac{1}{\sqrt{N}}\left[u^T_{n},\  -u^T_{n}\right]^T $ is no longer an eigenvector of graph $ G $ associated to the second smallest eigenvalue $ \mu_{N-1} $ and, therefore, the minimum cut cannot be obtiened by only cutting all the interconnection links. 
 
In summary, the physical meaning of $ p^* $ in terms of the minimum cut is that if $ p < p^* $, the minimum cut can be obtained by cutting all the interconnection links, whereas above the transition point, i.e., when $p> p^* $, the minimum cut involves both links within each subgraph and interdependent edges between the two subgraphs of the interconnected network $ G $.

\section{Exact threshold for special structures of interdependent networks}
\label{sec_exact_p*}
In this section, we analytically determine the structural threshold $ p^* $ for special graphs $ G_1 $ and $ G_2 $ or a special interconnection matrix $ B $. 
\subsection{Coupled identical circulant graphs}
Let $ x_{n-1} $ be the eigenvector associated to the second smallest eigenvalue $ \mu_{n-1}(Q_1) $ of the Laplacian matrix $ Q_1 $ of graph $ G_1 $. For vector $ x=\left(x_{n-1}^T,\ x_{n-1}^T\right)^T $ and $ Q_2=Q_1 $, the eigenvalue equation in (\ref{eq_laplacian_eigenvalue_equation}) reads
\begin{equation}
\label{eq_coupled_identical_circulant}
\begin{bmatrix}
Q_1+kpI & -p\hat{B}\\
-p\hat{B}^T & Q_1+kpI
\end{bmatrix}\begin{bmatrix}
x_{n-1}\\
x_{n-1}
\end{bmatrix}= \begin{bmatrix}
\mu_{n-1}(Q_1)x_{n-1}+kpx_{n-1}-p\hat{B}x_{n-1}\\
\mu_{n-1}(Q_1)x_{n-1}+kpx_{n-1}-p\hat{B}^Tx_{n-1}
\end{bmatrix}
\end{equation}
where $ \hat{B} $ is a zero-one matrix satisfying $ \hat{B}=\frac{B}{p} $. As mentioned in section \ref{sec_preliminary}, circulant matrices are commutative. If two matrices commute, the two matrices have the same set of eigenvectors \cite{PVM_graphspectra}. When $ Q_1 $ and $ \hat{B} $ are symmetric circulant matrices, $ Q_1 $ and $ \hat{B} $ commute, i.e., $ Q_1\hat{B}=\hat{B}Q_1 $, and the eigenvectors of $ Q_1 $ and $ \hat{B} $ are the same \cite{PVM_graphspectra}. The eigenvector $ x_{n-1} $ of the Laplacian $ Q_1 $ is also an eigenvector of matrix $ \hat{B} $ belonging to the eigenvalue $ \lambda  $, where $ \lambda =\frac{x_{n-1}^T\hat{B}x_{n-1}}{x_{n-1}^Tx_{n-1}}=2x_{n-1}^T\hat{B}x_{n-1}$ because the normalization $ x^Tx=2x_{n-1}^Tx_{n-1}=1 $. Substituting $ \hat{B}x_{n-1}=\lambda x_{n-1} $ in (\ref{eq_coupled_identical_circulant}) yields  
\begin{equation*}
\begin{bmatrix}
Q_1+kpI & -p\hat{B}\\
-p\hat{B}^T & Q_1+kpI
\end{bmatrix}\begin{bmatrix}
x_{n-1}\\
x_{n-1}
\end{bmatrix}= \left(\mu_{n-1}(Q_1)+kp-\lambda p\right)\begin{bmatrix}
x_{n-1}\\
x_{n-1}
\end{bmatrix}
\end{equation*}
The vector $ x=\left(x_{n-1}^T,\ x_{n-1}^T\right)^T $ is an eigenvector of $ Q $ associated to eigenvalue   
$ \mu=\mu_{n-1}(Q_1)+(k-\lambda)p $.

When the coupling weight $ p $ is small enough, the non-trivial eigenvalue $ \mu^*=2kp $ in (\ref{eq_nontrivial_eigenvalue}) could be the algebraic connectivity $ \mu_{N-1} $ and the eigenvalue $ \mu_{n-1}(Q_1)+(k-\lambda)p  $ could be made to be the third smallest eigenvalue $ \mu_{N-2} $. As already pointed out before, by increasing the coupling weight $ p $, a transition of the algebraic connectivity $ \mu_{N-1} $ occurs, where $ \mu^*=2kp $ is no longer the second smallest one. As the transition occurs at the point $ p^* $ such that $ 2kp^*=\mu_{n-1}(Q_1)+\left(k-\lambda\right)p^* $, one gets
\begin{equation*}
p^*= \frac{\mu_{n-1}}{k+\lambda}
\end{equation*}
where $ \lambda =2x_{n-1}^T\hat{B}x_{n-1}$.

Figure \ref{fig_transition_special_topology}(a) shows the algebraic connectivity of an interdependent network that consists of two identical circulant graphs with a $ 2 $-to-$ 2 $ ($ k=2 $) interconnection. The size of each circulant graph is $ n=100 $ with average degree $ d_{av}=6 $.  When the coupling strength $ p \leq p^* $, the algebraic connectivity $ \mu_{N-1} $ is $ 4p $. When $ p \geq p^*$, the algebraic connectivity in Figure \ref{fig_transition_special_topology}(a) is analytically expressed as $ \mu_{N-1}=\mu_{n-1}(Q_1)+(2-\lambda)p  $. The transition occurs at the point $ p^*= \frac{\mu_{n-1}}{2+\lambda} $, where $ \lambda =2x_{n-1}^T\hat{B}x_{n-1}$.
 
\begin{figure}[t!]
\centering
\includegraphics[width=0.8\textwidth]{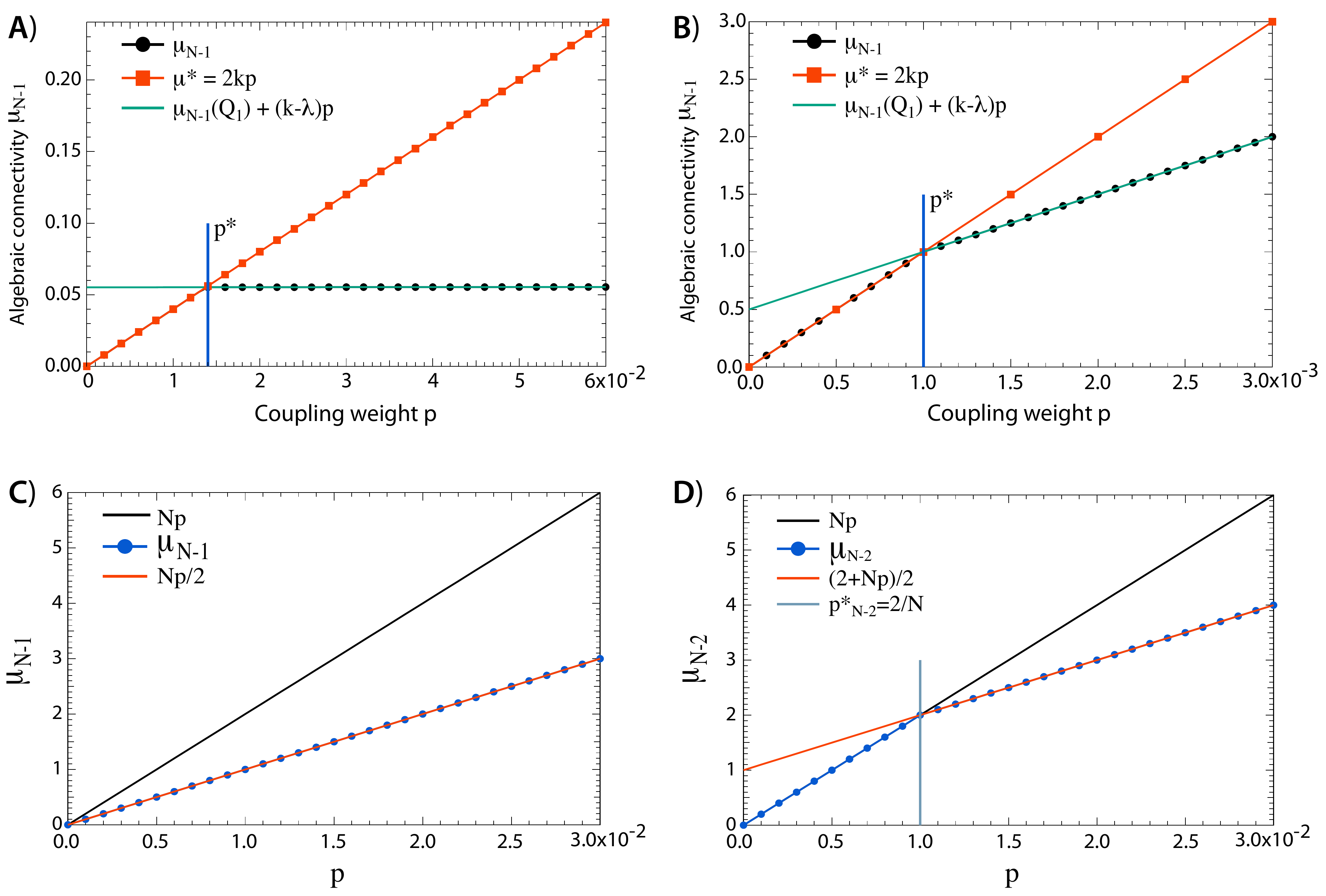}
\caption{Exact transition threshold for special structures including (a) coupled circulant graphs (b) fully coupled Erd\H{o}s-R\'{e}nyi graphs and (c),(d) fully coupled star graphs.}
\label{fig_transition_special_topology}
\end{figure} 

\subsection{$ n $-to-$ n $ interconnection}
\label{subsec_BequalsJ}
For an $ n $-to-$ n $ interconnection pattern, the Laplacian matrix of the interdependent graph $ G $ reads
\begin{equation*}
Q=\begin{bmatrix}
Q_1+pnI & -pJ_{n \times n}\\
-pJ_{n \times n} & Q_2+pnI
\end{bmatrix}
\end{equation*}
where the $ n \times n $ all-one matrix $ J $ represents that one node in graph $ G_1 $ connects to all nodes in graph $ G_2 $ and vice versa. Graph $ G $ is the join \cite{harary1969graph} of graphs $ G_1 $ and $ G_2 $ if the coupling weight $ p=1 $. 

Let $ x_{1} $ be the eigenvector associated to the eigenvalue $ \mu_{n-1}(Q_1) $ of graph $ G_1 $ and $  x_{2} $ be the eigenvector associated to the eigenvalue $ \mu_{n-1}(Q_2) $ of graph $ G_2 $. For vectors $ x=\left(x_{1}^T,\ 0\right)^T $ and $ x=\left(0,\ x_{2}^T \right)^T $, the eigenvalue equation for the Laplacian matrix $ Q $ of $ G $ can be written as

\begin{equation*}
\begin{bmatrix}
Q_1+pnI & -pJ\\
-pJ & Q_2+pnI
\end{bmatrix}\begin{bmatrix}
x_{1}\\
0
\end{bmatrix}=(\mu_{n-1}(Q_1)+np)\begin{bmatrix}
x_{1}\\
0
\end{bmatrix}
\end{equation*}
\begin{equation*}
\begin{bmatrix}
Q_1+pnI & -pJ\\
-pJ & Q_2+pnI
\end{bmatrix}\begin{bmatrix}
0 \\
x_{2}
\end{bmatrix}=(\mu_{n-1}(Q_2)+np)\begin{bmatrix}
0 \\
x_{2}
\end{bmatrix}
\end{equation*}
Similarly to the previous derivation, also for an $ n $-to-$ n $ ($ k=n $) interconnection, the non-trivial eigenvalue $ \mu^*=2np $ can be made equal to the algebraic connectivity $ \mu_{N-1}\left(Q\right) $ of the Laplacian $ Q $ if the coupling weight $ p $ is small. Moreover, the eigenvalue $ \text{min}\{\mu_{n-1}(Q_1),\ \mu_{n-1}(Q_2)\}+np $ could be the third smallest eigenvalue $ \mu_{N-2}\left(Q\right) $ for small values of $ p $. Taking into account that the transition threshold $ p^* $ occurs when $ \mu_{N-1}\left(Q\right)=\mu_{N-2}\left(Q\right) $, we get  
\begin{equation}
p*=\text{min}\left\lbrace\frac{\mu_{n-1}(Q_1)}{n},\ \frac{\mu_{n-1}(Q_2)}{n}\right\rbrace
\label{eq_p*_all_to_all_coupling}
\end{equation}

Figure \ref{fig_transition_special_topology}(b) shows the algebraic connectivity of the interdependent network consisting of two Erd\H{o}s-R\'{e}nyi graphs $ G_p(n) $ with $ n=500 $ nodes and average degree $ d_{av}=6 $. The interconnection pattern in this figure is $ n $-to-$ n $. Figure \ref{fig_transition_special_topology}(b) demonstrates that when the coupling weight $ p $ is small, the algebraic connectivity is $ \mu_{N-1}=2np $. With the increase of $ p $, the algebraic connectivity is described by $ \mu_{N-1}=\text{min}\{\mu_{n-1}(Q_1),\ \mu_{n-1}(Q_2)\}+np $. The transition occurs when $ 2np=\text{min}\{\mu_{n-1}(Q_1),\ \mu_{n-1}(Q_2)\}+np  $ and the threshold $ p^* $ obeys (\ref{eq_p*_all_to_all_coupling}).

\subsection{$ \left(n-1\right) $-to-$ \left(n-1\right) $ interconnection}

When $ B=p\left(J-I\right) $ and $ G_2=G_1 $, the eigenvalue equation for the Laplacian matrix $ Q $ reads, with vector $ x=\left(x_{n-1}^T,\ -x_{n-1}^T \right)^T $ where $ x_{n-1}  $ is an eigenvector associated to the algebraic connectivity $ \mu_{n-1}(Q_1) $ of graph $ G_1 $,  as

\begin{equation}
\begin{bmatrix}
Q_1+p(n-1)I & -p\left(J-I\right)\\
-p\left(J-I\right) & Q_1+p(n-1)I
\end{bmatrix}\begin{bmatrix}
x_{n-1}\\
-x_{n-1}
\end{bmatrix}=\left(\mu_{n-1}(Q_1)+\left(n-2\right)p\right)\begin{bmatrix}
x_{n-1}\\
-x_{n-1}
\end{bmatrix}
\end{equation}
The non-trivial eigenvalue follows $ \mu^*=2(n-1)p $ for an $ \left(n-1\right) $-to-$ \left(n-1\right) $ interconnection. When $ p $ is small, the eigenvalue $ 2(n-1)p $ could be made equal to $ \mu_{N-1} $ and the eigenvalue $ \mu_{n-1}(Q_1)+\left(n-2\right)p $ could be the third smallest eigenvalue $ \mu_{N-2} $. At the transition $ p=p^* $, we have that $\mu_{N-1}=\mu_{N-2} $ from which the threshold $ p^* $ follows as
\begin{equation*}
p*=\frac{\mu_{n-1}(Q_1)}{n}.
\end{equation*}

\subsection{A graph coupled with its complementary graph}
The complementary graph $ G_1^c $ of a graph $ G_1 $ has the same set of nodes as $ G_1 $ and two nodes are connected in $ G_1^c $ if they are not connected in $ G_1 $ and vice versa \cite{PVM_graphspectra}. The adjacency matrix of the complementary graph $ G_1^c $ is $ A_1^c=J-I-A_1 $. The Laplacian of the complementary graph $ G_1^c $ follows $ nI-J-Q_1 $. 

For an interdependent graph $ G $ consisting of a graph $ G_1 $ and its complementary graph $ G_1^c $ with an $ n $-to-$ n $ interconnection pattern, the Laplacian matrix $ Q $ of the interdependent graph $ G $ reads 
\begin{equation*}
Q=
\begin{bmatrix}
Q_1+npI & -pJ\\
-pJ & nI-J-Q_1+npI
\end{bmatrix}
\end{equation*}

Let $ x_{n-1} $ be the eigenvector associated to the eigenvalue $ \mu_{n-1} $ of the graph $ G_1 $ and $ x_1 $ be the eigenvector associated to the eigenvalue $ \mu_1 $. For vectors $ x=\left(x_{n-1}^T, \ 0\right)^T $ and $ x=\left(0, \ x_1^T\right)^T $, the eigenvalue equation for the Laplacian matrix $ Q $ of $ G $ can be written as
\begin{equation}
\begin{bmatrix}
Q_1+npI & -pJ\\
-pJ & nI-J-Q_1+npI
\end{bmatrix}\begin{bmatrix}
x_{n-1}\\
0
\end{bmatrix}=(\mu_{n-1}(Q_1)+np)\begin{bmatrix}
x_{n-1}\\
0
\end{bmatrix}
\label{eq_eigenvalue1_complementary}
\end{equation}
\begin{equation}
\begin{bmatrix}
Q_1+npI & -pJ\\
-pJ & nI-J-Q_1+npI
\end{bmatrix}\begin{bmatrix}
0 \\
x_{1}
\end{bmatrix}=(n+np-\mu_1(Q_1))\begin{bmatrix}
0 \\
x_{1}
\end{bmatrix}
\label{eq_eigenvalue2_complementary}
\end{equation}
Following the same procedure as in the previous examples, at the transition point we have that the equality $ \mu_{N-1}\left(Q\right)=\mu_{N-2}\left(Q\right) $ holds, which yields
\begin{equation*}
p*=\min\left(\frac{\mu_{n-1}(Q_1)}{n},\ 1-\frac{\mu_{1}(Q_1)}{n}\right)
\end{equation*}

\subsection{An example of the non-existence of the structural transition}
\label{sec_non_exist_p*}
In this subsection, we consider an interdependent network consisting of a star graph $ G_1 $ and its complementary graph $ G_1^c $ while the interconnection pattern is $ n $-to-$ n $. For a star graph with size $ n $, the eigenvalues of the Laplacian \cite{PVM_graphspectra} are $ 0 $, $ 1 $ with multiplicity $ n-2 $ and $ n $. 
Substituting $ \mu_{n-1}(Q_1)=1 $ and $ \mu_1(Q_1)=n  $ into eigenvalue equations (\ref{eq_eigenvalue1_complementary}) and (\ref{eq_eigenvalue2_complementary})  yields two eigenvalues $ np $ and $ np+1 $. 

When the coupling weight $ p > 0 $, the non-trivial eigenvalue $ \mu^*=2np $ cannot be the second smallest eigenvalue of the Laplacian $ Q $ because it is always larger than the eigenvalue $ np $. Hence, the transition between $ \mu^* $ and the algebraic connectivity $ \mu_{N-1}\left(Q\right) $ will never occur as shown in Figure \ref{fig_transition_special_topology}(c).
Instead, when $ p $ is small, the non-trivial eigenvalue $ \mu^*=2np $ can be made the third smallest eigenvalue $ \mu_{N-2}\left(Q\right) $. By increasing the coupling weight $ p $, the eigenvalue $ \mu^*=2np $ may no longer be the third smallest eigenvalue of the Laplacian $ Q $. There exists a threshold denoted as $ p^*_{N-2} $ such that $ \mu^*=2np $ exceeds $ \mu_{N-2}\left(Q\right) $ when $ p > p^*_{N-2} $. When $ p \leq  p^*_{N-2} $ then the third smallest eigenvalue follows $ \mu_{N-2}\left(Q\right)=2np $. Above the transition point $ p^*_{N-2}  $, the non-trivial eigenvalue $ \mu^*=2np$ exceeds eigenvalue $ 1+np $. The transition occurs when $ 2np*=1+np* $ resulting in $  p^*_{N-2} =\frac{1}{n} $. Figure \ref{fig_transition_special_topology}(d) shows that the transition occurs at the point $  p^*_{N-2} =\frac{1}{n} $.

Note, however, that in the above example, the complementary graph $ G_1^c $ of a star is a disconnected graph. The hub node in the star $ G_1 $ is an isolated node in graph $ G_1^c $. The coupling is stronger between graph $ G_1 $ and the connected component in graph $ G_1^c $ than that between graph $ G_1 $ and the isolated node in $ G_1^c $. The isolated node first decouples from the interdependent network $ G $ before the connected component in $ G_1^c $ decouples from the interdependent graph $ G $. As a result, the structural transition in $ p $ occurs at the third smallest eigenvalue rather than at the second smallest eigenvalue. The above example also agrees with the upper bound in (\ref{eq_upper_bound_treshold}) in that the threshold $ p^*=0 $ when $ \mu_{n-1}(Q_1)=0 $ or $ \mu_{n-1}(Q_2)=0 $. There is no transition between the non-trivial eigenvalue $ \mu^*=2kp $ and the algebraic connectivity $ \mu_{N-1} $, if one of the coupled graphs is disconnected. 

\section{Conclusion}
\label{sec_conclusion}
In this paper, we have studied the structural transition of interdependent networks. We first generalized the one-to-one interconnection coupling to a general $ k $-to-$ k $ inter-coupling scheme for interdependent networks. This representation of the couplings between the networks that made up the whole system is more realistic and could represent more situations of practical interest. However, we acknowledge that the interconnection matrix $ B $ representing the $ k $-to-$ k $ interconnection obeys regularity (constant row and column sum), which still represents a simplification of real systems. Nonetheless, the more complex scenario addressed here allows to deduce the non-trivial eigenvalue of such interdependent networks \cite{van2016interconnectivity}. 

For the general $ k $-to-$ k $ interconnection ($ B \neq pI $ unless $ k=1 $) studied throughout this paper, a number of new results and properties of the transition threshold $ p^* $ can be readily obtained. For connected graphs $ G_1 $ and $ G_2 $, we showed that the transition threshold $ p^* $ is upper bounded by the minimum algebraic connectivity of graphs $ G_1 $ and $ G_2 $ divided by $ k $. Additionally, we have shown that networks that are divisible to regularly interconnected subgraphs, show a transition threshold that is determined from the coarse-grained or quotient graph. These results could be important in some applications. For instance, the bounds and the exact value of the transition threshold $ p^* $ can be used to identify the interactions and the multi-layer coupling pattern of neural networks, as they have been suggested to operate, in a healthy human brain, around the transition point \cite{Tewarie2016integrating}. Our physical interpretation of the threshold $ p^* $ is also of interest. Namely, we have argued that below the transition threshold $ p^* $, the minimum cut of the network includes all the interconnection links, whereas above it, the minimum cut contains both the interconnection links between graphs $ G_1 $ and $ G_2 $ and the links within $ G_1 $ and $ G_2 $. Finally, we have derived exact expressions for the threshold $ p^* $ for some special topologies, and shown that if one of the graphs $ G_1 $ or $ G_2 $ is disconnected, then the structural threshold $ p^* $ for the algebraic connectivity does not exist. Altogether, our results allow further advances into the theory of multilayer networks and could pave the way to similar studies that consider more realistic networks.

\section*{Acknowledgement}
We are very grateful to Caterina Scoglio for valuable discussions. This research is supported by the China Scholarship Council (CSC). YM acknowledges partial support from the Government of Arag\'on, Spain through grant E36-17R, and by MINECO and FEDER funds (grant FIS2017-87519-P). 

\appendix
\section{The quotient graph preserves the algebraic connectivity under certain conditions for a particular class of graphs}
The Cartesian product $ G_1 \square G_2 $ of two graphs $ G_1 $ and $ G_2 $ with node set $ \mathcal{N}_1 $ and $ \mathcal{N}_2 $ is a graph such that (i) the node set of $ G_1 \square G_2 $ is $ \mathcal{N}_1 \times \mathcal{N}_2   $ and (ii) two nodes $ i_1i_2 $ and $ j_1j_2 $ are connected in $ G_1 \square G_2 $ if either $ i_1=j_1 $ and $ i_2 $ is connected to $ j_2 $ or $ i_2=j_2 $ and $ i_1 $ is connected to $ j_1 $. If all the subgraphs in (\ref{eq_A_1}) are identical and the matrix $ R $ is the identity matrix, the two-layer coupled network can be written as the Cartesian product of the subgraph $ H_1^{\left(1\right)} $ fully connected with subgraph $ H_1^{\left(2\right)} $ and the path graph $ P_m $ with $ m = \frac{n}{k} $ nodes. The corresponding quotient graph can be obtained by the Cartesian product of a path graph with $ 2 $ (number of layers) nodes weighted by $ kp $ and a path graph with $m $ nodes. The adjacency matrix of the original graph and the quotient graph thus follows  
\begin{equation}\label{key}
A = 
\begin{bmatrix}
A_{H_1^{\left(1\right)}} & pJ_{k \times k} \\
pJ_{k \times k}  & A_{H_1^{\left(2\right)}} 
\end{bmatrix}
\square\
A_{P_m}
\end{equation}
and
\begin{equation}\label{key}
\widetilde{A}= 
\begin{bmatrix}
0 & kp \\
kp & 0 
\end{bmatrix}
\square\
A_{P_m}
\end{equation}
where $ \square $ represents the Cartesian product. The example topology in Figure \ref{fig_topology_exact_threshold} can be obtained from Cartesian product as shown in Figure \ref{fig_topology_from_cartesian_product}.
\begin{figure}[t!]
\centering
\includegraphics[width=0.7\textwidth]{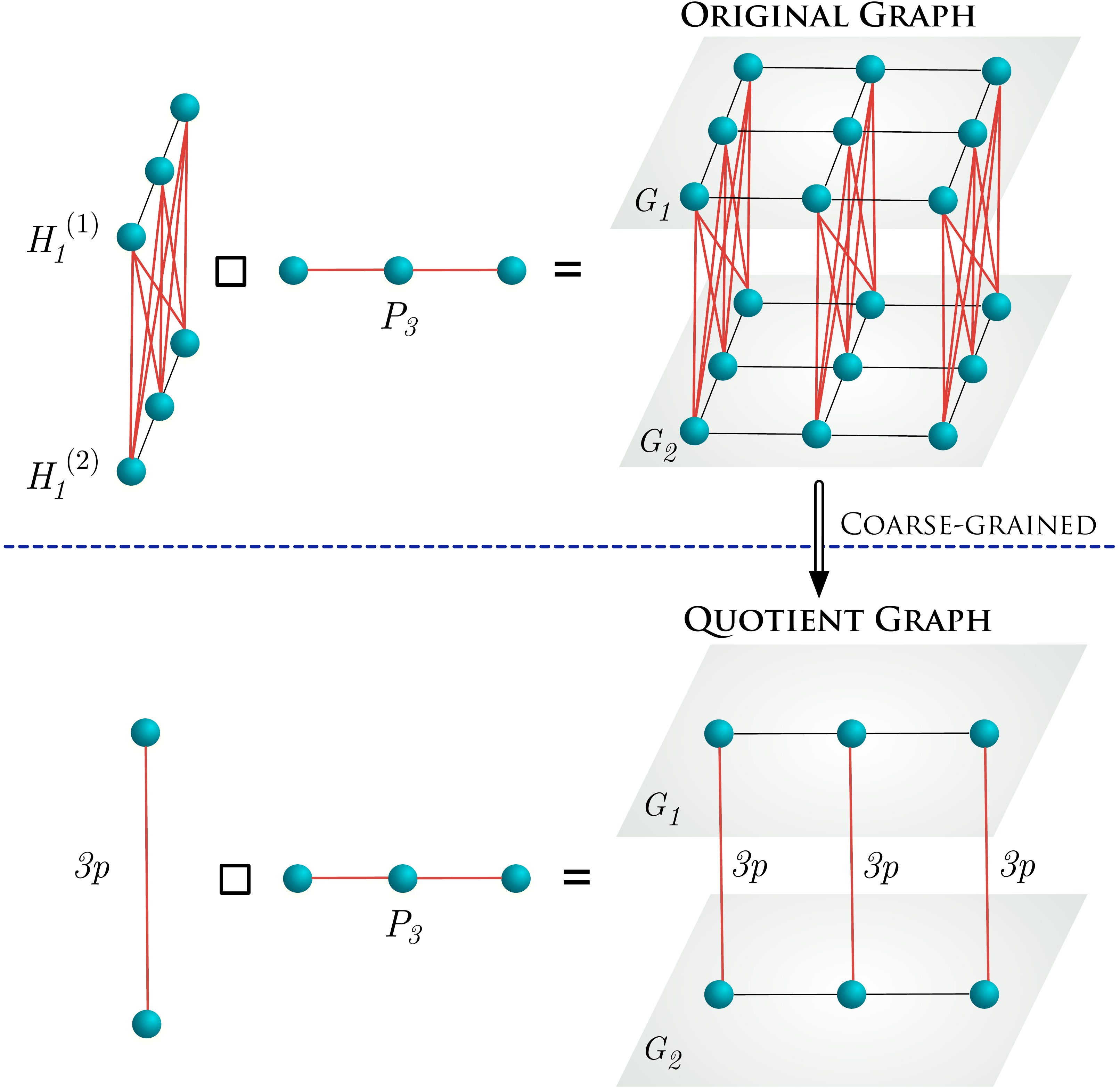}
\caption{The figure illustrates how the example topology in Figure \ref{fig_topology_exact_threshold} can be obtained from the Cartesian product of subgraphs.}
\label{fig_topology_from_cartesian_product}
\end{figure}

If two graphs $ G_1 $ and $ G_2 $ have Laplacian eigenvalues $ \mu\left(G_1\right) $ and $ \mu\left(G_2\right) $, then the Laplacian eigenvalues of the Cartesian product \cite{brouwer2011spectra} of $ G_1 $ and $ G_2 $ are $ \mu\left(G_1\right) + \mu\left(G_2\right) $. Thus, the Laplacian eigenvalues of the quotient graph are those of the path graph $ P_m $ plus those of $ P_2 $ with weight $ kp $. The second smallest eigenvalue can be either $ 2kp $ of the path $ P_2 $ or $ \mu_{m-1}\left(P_m\right) $ of the path $ P_m $. Similarly, the Laplacian eigenvalues of the original graph are those of path graph $ P_m $ plus those of the fully connected subgraphs $ H_1^{(1)} $ and $ H_1^{(2)}  $. The second smallest eigenvalue can be either $ 2kp $ of the fully connected $ H_1^{(1)} $ and $ H_1^{(2)}  $ or $ \mu_{m-1}\left(P_m\right) $ of the path $ P_m $. For this particular class of graphs, the quotient graph preserves the algebraic connectivity, which follows $ \min\{2kp,\ \mu_{m-1}\left(P_m\right)\} $, of the original graph.  

\end{document}